\documentclass[%
 reprint
]{revtex4-2}
\usepackage{iceberg}

\begin{document}

\title{2 Fast 2 Surgery: Fast surgery on \qcode{QLDPC} codes with $\tilde{O}(n(k + d))$ space overhead}

\author{Nou\'edyn Baspin}
\affiliation{%
 Iceberg Quantum\\
 Sydney, Australia
}%
\author{{\scriptsize\texttt{nouedyn@iceberg-quantum.com}}}
\newcommand{\measured}{\mathop{\mathfrak{l}}}
\newcommand{\lpreimage}{\mathop{L}}
\newcommand{\tester}{\mathop{H_{\lpreimage}}}

\begin{abstract}
We show that the fast surgery framework of \cite{baspin2025fastsurgeryquantumldpc} can be instantiated using at most $\tilde{O}(n(k + d))$ ancilla qubits, improving on the previously achievable $\tilde{O}(nkd)$ overhead. This result also establishes that fast surgery can be performed at a cost comparable to the $\tilde{O}(nk)$ overhead required for regular surgery \cite{cowtan2025parallel}. We build the surgery ancilla by first starting with a dense, but expanding system; and then applying expansion-preserving sparsification methods from gauging \cite{williamson2024gauging} and code embeddings \cite{yuan2026unified, yuan2026quantumweightreductionlayer}.
\end{abstract}

\maketitle

\section{Intro}

The prospect of practical quantum computers is hampered by the operational costs they are expected to incur \cite{campbell2017roads}. 
Compared to their classical counterparts, state-of-the-art quantum computers consume significant space and time resources in order to run even the simplest algorithms \cite{Litinski2019, gidney2021twenty}.
The necessity to reduce this footprint has thus shaped much of the current QEC landscape \cite{Bravyi2024,xu2024constant,gidney2025factor2048bitrsa, webster2026pinnacle, cain2026tenthousand}.

Space overheads are generally attributable to the type of codes used \cite{gottesman2014constant,cohen2022surgery, nguyen2025constantspace,bhardwaj2026highrate}. 
For a given scheme, it is conventionally written that one uses $n$ physical qubits to encode $k$ logical qubits. 
Naturally, the smaller the ratio $n/k$ is, the smaller the space overhead. 
The need for code families satisfying a constant $n/k$, for growing $k$, spurred the early developments \cite{freedman2002systolic, kim2007hurwitz, zemor2009cayley, delfosse2013tradeoffs, tillich2014hypergraph, guth2014quantum} of \qcode{QLDPC} codes \cite{breuckmann2021balanced}. 
It can be useful to note that such a quest has a natural horizon: once constant-rate families are discovered, further improvements can only reduce the space overhead by a constant factor.

Conversely, time overheads emerge from inefficiencies in the handling of the data. As codes are designed to insulate the logical information they protect from the exterior environment, the best codes will also make it harder to operate said information. This can be coarsely quantified by $d$, the number of errors a code can undergo without failing; and aside from a few particularly well behaved exceptions \cite{breuckmann2024fold, eberhardt2024fold}, performing a gate will typically require $\Omega(d)$ time \cite{Litinski2019}. 

Among existing methods to perform such logical gates, lattice surgery \cite{bombin2009measurements, horsman2012surgery} has emerged as a prominent component of low-overhead architectures due to its ability to generalise to \qcode{QLDPC} codes \cite{cohen2022surgery, williamson2024gauging, cowtan2025parallel}.
Despite surgery also suffering from a $\Omega(d)$ delay in its original formulation, recent results have shown that it can be made fast \cite{baspin2025fastsurgeryquantumldpc, cowtan2025fast} -- i.e. performed in $O(1)$-time -- generally at the cost of an increased space footprint. While in some special cases constant space overhead can still be achieved \cite{chang2026constanttime}, the most general construction for fast surgery requires an additional $\tilde{O}(nkd)$ physical qubits \cite{baspin2025fastsurgeryquantumldpc}.

In this work we show that fast surgery can be performed with $\tilde{O}(n(k + d))$ physical qubits, improving on $\tilde{O}(nkd)$, and bringing the cost in line with the $\tilde{O}(nk)$ overhead of regular surgery \cite{cowtan2025parallel}. Our main result is \Cref{theorem:main}, which constructs a \qcode{QLDPC} surgery ancilla system with large enough expansion and co-systole to satisfy the requirements for fast surgery in \cite{baspin2025fastsurgeryquantumldpc}.

\begin{restatable}{theorem}{mainthm}
\label{theorem:main}
    Let $C_\bullet$ be a \qcode{QLDPC} code, and $\measured \subset H_1(C)$ an arbitrary subspace of logicals of dimension $\kappa = \dim \measured$. Then there exists $D^{\text{sparse}}_\bullet$, $\gamma_\bullet^{\text{sparse}}: D^{\text{sparse}}_\bullet \to C_\bullet$:

    \begin{equation}
        \begin{tikzcd}[ampersand replacement=\&]
        	{D_2^{\text{sparse}}} \& {D_1^{\text{sparse}}} \& {D_0^{\text{sparse}}} \\
        	{C_2} \& {C_1} \& {C_0}
        	\arrow[from=1-1, to=1-2]
        	\arrow["{\gamma_2^{\text{sparse}}}", from=1-1, to=2-1]
        	\arrow[from=1-2, to=1-3]
        	\arrow["{\gamma_1^{\text{sparse}}}", from=1-2, to=2-2]
        	\arrow["{\gamma_0^{\text{sparse}}}", from=1-3, to=2-3]
        	\arrow[from=2-1, to=2-2]
        	\arrow[from=2-2, to=2-3]
        \end{tikzcd}
    \end{equation}
    with the following properties
    \begin{enumerate}
        \item $D^{\text{sparse}}$ is $1$-expanding relative to $\gamma_1^{\text{sparse}}$
        \item $\im [\gamma_1^{\text{sparse}}] = \measured$
        \item $\syst^1(D^{\text{sparse}}) \geq \syst^1(C)$ 
        \item $\partial_1^{D^{\text{sparse}}}$ and $\gamma_\bullet^{\text{sparse}}$ are $O(1)$-sparse
        \item $\dim D_0^{\text{sparse}} \leq \tilde{O}(n(k-\kappa + d))$, where $n = \dim C_1$
    \end{enumerate}

\end{restatable}

The primary practical limitation of our construction is that $\partial_2^{D^{\text{sparse}}}$, corresponding to the metachecks, is not necessarily sparse. 

\section{Prelimaries}

An $m$-dimensional chain complex, $(C_\bullet, \partial_\bullet)$, over $\F_2$ is a collection of $\F_2$-vector spaces, $\left\{ C_i \right\}_{i=0}^{m}$, and linear boundary maps $\left\{ \partial_i: C_{i+1} \to C_i \right\}_{i=1}^{m}$
\begin{equation}
    C_\bullet =
    \begin{tikzcd}
        C_D \arrow[r, "\partial_D"] & C_{D-1} \arrow[r, "\partial_{D-1}"] & \cdots \arrow[r, "\partial_2"] & C_1 \arrow[r, "\partial_1"] & C_0
    \end{tikzcd}\, ,
\end{equation}
such that the boundary maps satisfy $\partial_{i-1} \circ \partial_{i} = 0$. 
We assume that each vector space $C_i$ comes with a distinguished basis and its elements are called $i$-chains.
Moreover, the elements of $B_i = \im \partial_{i+1}$ are called $i$-boundaries, and the elements of $Z_i = \ker \partial_i$ are called $i$-cycles. 
The homology groups are defined as $H_i\left(C_\bullet\right)$, i.e., they correspond to cycles that are not boundaries
\begin{equation}
    H_i\left(C_\bullet\right) = \ker \partial_i / \im \partial_{i+1}\, .
\end{equation}
We can similarly define the cohomology groups, which for chain complexes over $\F_2$-vector spaces are defined as
\begin{equation}
    H^i\left(C^{\bullet}\right) = \ker \delta_i / \im \delta_{i-1}\, ,
\end{equation}
where $\delta_i = \partial_{i+1}^T$ for the considered $\F_2$ vector spaces. 
Since the vector spaces are equipped with distinguished bases, we define a norm $\left \| \cdot \right \|$ on the $i$-chains given by the Hamming weight. 
Then the $i$-systole and $i$-cosystole are defined by
\begin{align}
    \text{Syst}_i\left(C_{\bullet}\right) &= \min \left\{ \left\| v \right\| \mid v \in Z_i \setminus B_i \right\}, \\
    \text{Syst}^i\left(C^{\bullet}\right) &= \min \left\{ \left\| v \right\| \mid v \in Z^i \setminus B^i\right\}.
\end{align}
We say a linear map $A$ is $\omega$-bounded if $\norm{Av} \leq \omega \norm{v}$ for any $v$. In a finite dimensional space this corresponds to $A$ having column weight at most $\omega$. We further say that $A$ is $\omega$-sparse if both $A$ and $A^T$ are $\omega$-bounded.

\subsection*{Maps, Cones, and Fibers}
There exists a natural notion of a morphism for chain complexes. 
Suppose we have two chain complexes, $\left(C_\bullet, \partial^C\right)$ and $\left(D_\bullet, \partial^D\right)$, along with a set of homomorphisms $\gamma_i: D_i \to C_i$. 
If the homomorphisms and the boundary maps satisfy the commutativity condition
\begin{equation}
    \gamma \partial = \partial \gamma \, ,
\end{equation}
then the map $\gamma_i$ describes a chain map between the complexes $C_\bullet$ and $D_\bullet$, which we refer to as \emph{homomorphic}.
We can represent this visually via the following commuting diagram
\begin{equation}
    \begin{tikzcd}
        \cdots \arrow[r] & D_{i+1} \arrow[r] \arrow[d, "\gamma_{i+1}"] & D_{i} \arrow[r] \arrow[d, "\gamma_{i}"] & D_{i-1} \arrow[r] \arrow[d, "\gamma_{i-1}"]& \cdots \\
        \cdots \arrow[r] & C_{i+1}  \arrow[r] & C_{i}  \arrow[r] & C_{i-1}  \arrow[r] & \cdots
    \end{tikzcd} \, .
\end{equation}
We can define a new chain complex from the above complex. 

\begin{definition}
    Let $(C_\bullet, \partial^C), (D_\bullet, \partial^D)$ be $m$-dimensional chain complexes, and $\gamma: D_\bullet \to C_\bullet$ a chain map. Then the mapping cone $\cone(\gamma)$ of $\gamma$ is defined as:
    \begin{gather}
        \cone(\gamma)_i = C_i \oplus D_{i-1} \\
        \partial^{\cone(\gamma)}_i (c, d) = (\partial_i^C c + \gamma_{i-1} d, \partial^D_{i-1} d) 
    \end{gather}
    Now, for a chain map $\tau_\bullet: C_\bullet \to D_\bullet$, we can define the mapping fiber $\fib(\tau)$ of $\tau$ as:
    \begin{gather}
        \fib(\tau)_i = D_{i+1} \oplus C_i \\
        \partial^{\fib(\tau)}_i (d,c) = (\partial^D_{i+1} d + \tau_i c, \partial^C_i c)
    \end{gather}
\end{definition}

\begin{remark}
    The difference between $\cone(\gamma)$ and $\fib(\gamma)$ is only a matter of grading. In the first, $\cone(\gamma_i)$ is indexed in terms of the "left" component $C_i$, while $\fib(\gamma)_i$ would be indexed relative to the "right" component $D_i$.
\end{remark}

\subsection*{Systolic expansion}
There is a notion of expansion in chain complexes which generalizes the notion of edge expansion in graphs, which can be viewed as $1$-dimensional chain complexes.
\begin{definition}[Systolic expansion \cite{gromov2010singularities, kaufman2016isoperimetric, evra2016bounded}] \label{def:systolic_expansion}
    Let $D_\bullet$ be an $m$-dimensional complex and let $\epsilon > 0$. 
    For $i < m$ we say that the vector space $D_i$ is $\epsilon$-systolic expanding if
    \begin{gather}
        \forall f \in D_i, \quad \norm{\partial_i f}  \geq \epsilon \cdot  \dist(f, Z_i) \\
        \dist(f, Z_i) = \min_{z \in Z_i} \norm{f + z}
    \end{gather}
\end{definition}
A chain homomorphism from $D_\bullet$ to $C_\bullet$ naturally gives rise to a notion of \emph{relative} expansion. 
\begin{definition}[Relative expansion \cite{baspin2025fastsurgeryquantumldpc}]
Let $D_\bullet$ be an $m$-dimensional complex, $\gamma : D_\bullet \rightarrow C_\bullet$ a chain map, and let $\epsilon > 0$. 
    For $i < m$ we say that $D$ is $\epsilon$-systolic expanding relative to $\gamma_i$ if for all $f \in D_i$:
    \begin{equation}
        \norm{\partial_i f}  \geq \epsilon  \min \left \{  \dist(\gamma_i (f), \gamma_i  (Z_i(D))), \syst_i(C) \right \}
    \end{equation}
\end{definition}

\section{Main results}

Fundamentally, constructing a viable surgery ancilla is a problem about \emph{map sparsification}. For the sake of notation, let us fix
\begin{equation}
    C_\bullet = C_2 \to C_1 \to C_0, \quad \measured \subseteq H_1(C)
\end{equation}
Where $C_\bullet$ corresponds to the code of interest, and $\measured$ corresponds to the subspace of logicals we want to measure. Then the measurement could be simply realised by picking a lift $\iota: \measured \to C_1$, and the merged code would then be described by $\cone(\iota)$ \cite{Ide2025}.
\begin{equation}
    \begin{tikzcd}
	& {\measured} & \\
	{C_2} & {C_1} & {C_0}
	\arrow["\iota", from=1-2, to=2-2]
	\arrow[from=2-1, to=2-2]
	\arrow[from=2-2, to=2-3]
\end{tikzcd}
\end{equation}
However, crucially, the map $\iota$ is generally non-sparse, and yields an undesirable high weight ancilla. In order to correct for this, one can then try to sparsify $\iota$; this is the approach adopted in \cite{hastings2023quantumweightreduction, williamson2024layer, williamson2024gauging}, where $\iota$ takes the interpretation of additional stabilisers added to $C_\bullet$. 
Formally \cite{yuan2026unified, yuan2026quantumweightreductionlayer}, this sparsification replaces the ancilla $\measured$ by a sparse complex $D_\bullet^{\measured}$, and a sparse chain map $\gamma^{\measured}: D_\bullet^{\measured} \to  C_\bullet$, such that the maps $[\gamma_1^{\measured}]: H_1(D^{\measured}) \to H_1(C)$, $[\iota]: \measured \to H_1(C)$ satisfy:
\begin{equation}
    \im [\gamma_1^{\measured}] = \im [\iota]
\end{equation}
And the merged cone is then described as $\cone(\gamma^{\measured})$:
\begin{equation}
    \begin{tikzcd}
	{D^{\measured}_2} & {D^{\measured}_1} & {D^{\measured}_0} \\
	{C_2} & {C_1} & {C_0}
	\arrow[from=1-1, to=1-2]
	\arrow[from=1-1, to=2-1]
	\arrow[from=1-2, to=1-3]
	\arrow[from=1-2, to=2-2]
	\arrow[from=1-3, to=2-3]
	\arrow[from=2-1, to=2-2]
	\arrow[from=2-2, to=2-3]
\end{tikzcd}
\end{equation}

For fast surgery in particular \cite{baspin2025fastsurgeryquantumldpc}, the ancilla $D_\bullet$ (we drop the $\measured$ superscript) is itself generally built as a mapping fiber $D_\bullet = \fib(\tau)$ of some map $\tau: C_\bullet \to E_\bullet$.

 This structure is chosen as it is particularly amenable to preserving some desirable properties from $C_\bullet$, like its co-systole. Once that requirement is satisfied, the most notable remaining difficulty in building a suitable $D_\bullet$ such that it also exhibits (relative) systolic expansion \cite{baspin2025fastsurgeryquantumldpc}. 
 
 In order to resolve the problem of producing expansion, we follow an approach inverse to \cite{baspin2025fastsurgeryquantumldpc}: we begin with a (dense) ancilla system $D_\bullet^{\text{dense}}$ that immediately exhibits expansion.
\begin{equation}
    D_\bullet^{\text{dense}} = \fib\left(\begin{tikzcd}[ampersand replacement=\&]
    	{C_2} \& {C_1} \& {C_0} \\
    	\& T
    	\arrow[from=1-1, to=1-2]
    	\arrow[from=1-2, to=1-3]
    	\arrow["\sigma", from=1-2, to=2-2]
    \end{tikzcd}\right)
\end{equation}
For some $\sigma$ we describe in the following. 
 Write $\pi : Z_1 \to H_1(C)$ the projection from the cycles to the corresponding homology group. Then we define
\begin{equation}
    \lpreimage = \pi^{-1} \measured \subset Z_1(C)
\end{equation}
Writing $\kappa = \dim \measured$, in general ${\lpreimage}$ can be thought of as the span of $\{l_1, \dots, l_{\kappa}\}$, a selected basis for $\measured$; plus $B_1(C)$. As a result we have $\dim {\lpreimage} = \kappa + r_X$, with $r_X  = \dim B_1(C)$. The space ${\lpreimage}$ should be understood as the target image of $\gamma_1 \ker \partial_1^D$. 

From here on we write $q = k - \kappa$ for readability. Note that $\measured^\perp \subset H^1(C)$ has exactly dimension $q$. Then we can pick $\{l^\star_{1}, l^\star_{2}, \dots, l^\star_{q}\}$ a basis of co-cycles corresponding to $\measured^\perp$ as a natural choice for $\sigma_{\text{cocycles}}$ is as follows:
\begin{equation}
    \sigma_{\text{cocycles}} = 
    \begin{bmatrix}
        l^\star_{1} \\
        l^\star_{2} \\
        \vdots \\
        l^\star_{q}
    \end{bmatrix}
\end{equation}
Indeed $\partial_1$ eliminates all non-cycles, while $\sigma_{\text{cocycles}}$ eliminates those cycles that are not in $\lpreimage$. And we therefore get:
\begin{equation}
    \ker \sigma_{\text{cocycles}} \cap \ker \partial_1 = L
\end{equation}
However, we are still lacking expansion. We will remediate that issue by introducing a second matrix $\sigma_{\text{boost}}$ such that:
\begin{equation}
    \norm{\sigma_{\text{boost}} u} + \norm{\sigma_{\text{cocycles}} u} +\norm{\partial_1 u} \geq d: u \not\in \lpreimage
\end{equation}

This condition is in fact much stronger than the relative expansion we need. We will then focus on the subset of chains that needs boosting
\begin{equation}
    S_d = \{ c \in \im \begin{bmatrix}
        \partial_1 & \sigma_{\text{cocycle}}
    \end{bmatrix}: 0< \norm{c} < d\}
\end{equation}
Those syndrome have to be boosted, which we will do through using random codes. From \Cref{lemma:random-coding-boosting} with $m = q + m_Z \leq O(n)$, and $M_Z = \dim C_0 \leq O(n)$ we get 
\begin{align}
    \sigma_{\text{boost}}&: C_1 \to \mathbb{F}_2^{t}, \quad t \leq O(d\log(n))\\
    \sigma_{\text{boost}} &= R\begin{bmatrix}
        H_{\text{logical}} \\
        \partial_1
    \end{bmatrix} \\
     \norm{\sigma_{\text{boost}} u} &+ \norm{\sigma_{\text{cocycles}} u} +\norm{\partial_1 u} \geq d: u \not\in \lpreimage
\end{align}

Note that the resulting ancilla satisfies a stronger version of expansion, which we call the $d$-threshold property.

\begin{definition}[$d$-threshold property]
Let $D_\bullet$ be an $m$-dimensional complex, $\gamma : D_\bullet \rightarrow C_\bullet$ a chain map. Then, for $i < m$ we say that $D$ has the $d$-threshold property relative to $\gamma_i$ if:
    \begin{align}
        \forall x \in D_i: \gamma_i(x) \not\in \gamma_i(Z_1(D)), \quad \norm{\partial_1^D x} \geq d
    \end{align}
\end{definition}

From these observations, we obtain a natural choice for $\sigma$ with the expected properties.

\begin{lemma}[Dense expanding map]
\label{lemma:expanding-pcm} 

Let $C_\bullet$ be a code, and $\measured \subset H_1(C)$ a subspace of logicals. Then there exists $\sigma : C_1 \to \F_2^s$ such that:

\begin{enumerate}
    \item There is a $1$-sparse chain map $\gamma_\bullet: \fib(\sigma) \to C_\bullet$ such that $\fib(\sigma)$ has the $\syst_1(C)$-threshold property relative to $\gamma_1$. 
    \item $\im [\gamma_1] = \measured$
    \item $\syst^1(\fib(\sigma)) \geq \syst^1(C)$
    \item $s \leq O(q+d\log(n))$
\end{enumerate}

\end{lemma}

\begin{proof}
    The previous conversation proves most points for $\sigma$ taken as follows: 
    \begin{equation}
        \sigma = \begin{bmatrix}
            \sigma_{\text{cocycles}} \\ \sigma_{\text{boost}}
        \end{bmatrix}, \quad \sigma: \F_2^n \to \F_2^s \cong \F_2^q \oplus \F_2^t
    \end{equation}
    And we take $\gamma_\bullet$ as follows:
    \begin{equation}
        \begin{tikzcd}[ampersand replacement=\&]
	\& {C_2} \& {C_1} \& {C_0} \\
	\&\& {\mathbb{F}_2^s } \\
	{C_2} \& {C_1} \& {C_0}
	\arrow[from=1-2, to=1-3]
	\arrow["{\gamma_2} =  \text{id}_{C_2}"{pos=0.7}, from=1-2, to=3-1]
	\arrow[from=1-3, to=1-4]
	\arrow["{\sigma}", from=1-3, to=2-3]
	\arrow["{\gamma_1} =  \text{id}_{C_1}"{pos=0.7}, from=1-3, to=3-2]
	\arrow["{\gamma_0} =  \text{id}_{C_0}"{pos=0.7}, from=1-4, to=3-3]
	\arrow[from=3-1, to=3-2]
	\arrow[from=3-2, to=3-3]
\end{tikzcd}
    \end{equation}

    The remaining point to prove is that of the co-systole. 
    While mapping fibers spontaneously preserve systoles, co-systoles can seem a bit trickier to preserve due to the addition of new redundancies through $\sigma^T$. 
    Fortunately, the structure of these redundancies is well controlled, as $\im \sigma^T = (\ker \sigma)^\perp = \lpreimage^\perp$, which should generally be understood as being constituted of co-boundaries, and some non-trivial co-cycles commuting with $L$. 
    
    We will see that the fact that $\sigma \partial_2 = 0$ (which is necessary for $\fib(\sigma)$ to be a valid chain complex) strongly restricts the systole.
    We have:
    \begin{align}
        Z^1(\fib(\sigma)) &= Z^1(C) \\
        B^1(\fib(\sigma)) &= B^1(C) + \im \sigma^T
    \end{align}
    Which, for the sake of convenience, we re-write as:
    \begin{align}
         B^1(C) + \im \sigma^T &= B^1(C) + (\ker \sigma)^\perp\\
    \end{align}
    From $\sigma \partial_2 = 0$, we have $B_1(C) \subset \ker \sigma$, and therefore $(\ker \sigma)^\perp \subset Z^1(C)$. This gives:
    \begin{align}
        \syst^1(\fib(\sigma)) &= \min_{ l \in Z^1(\fib(\sigma)) \setminus B^1(\fib(\sigma))}\norm{l }\\ 
        & = \min_{ l \in Z^1(C) \setminus B^1(C) + (\ker \sigma)^\perp}\norm{l } \\
        & \geq \min_{ l \in Z^1(C) \setminus B^1(C) }\norm{l } \\
        & = \syst^1(C)
    \end{align}
\end{proof}


However, $\sigma$ as chosen here is a dense matrix, and we are again reduced to having to sparsfiy the new stabilisers given by the rows of $\sigma$. 

This new starting point might seem disappointing as we had started with an equally dense matrix $\iota$. We highlight here that the actual contribution of that step is that sparsifying $\iota$ requires coming up with an expander complex, while here sparsifying $\sigma$ merely requires preserving the already present expansion.

We then show that the weight-reducing gauging gadgets of \cite{williamson2024gauging} are indeed expansion-preserving. We then use a dual version of those to build a similar expansion-preserving degree-reducing gadget.


\begin{remark}
    It is a question of general curiosity to understand how natural it is to base the ancilla $D_\bullet$ off a copy of $C_\bullet$. Interestingly, the co-systole preservation exhibited by \Cref{lemma:expanding-pcm} relies precisely on $D_2 = C_2$.
\end{remark}

\subsection{Main theorem}

\mainthm*

\begin{proof}
    We start by $D^{\text{dense}}_\bullet = \fib(\sigma)$, given by \Cref{lemma:expanding-pcm}. We obtain the following $1$-sparse chain map $\gamma_\bullet: \fib(\sigma) \to C_\bullet$:
    \begin{equation}
        \begin{tikzcd}
        	& {C_2} & {C_1} & {C_0} \\
        	&& T \\
        	{C_2} & {C_1} & {C_0}
        	\arrow[from=1-2, to=1-3]
        	\arrow["{\gamma_2}"{pos=0.8}, from=1-2, to=3-1]
        	\arrow[from=1-3, to=1-4]
        	\arrow["H", from=1-3, to=2-3]
        	\arrow["{\gamma_1}"{pos=0.8}, from=1-3, to=3-2]
        	\arrow["{\gamma_0}"{pos=0.8}, from=1-4, to=3-3]
        	\arrow[from=3-1, to=3-2]
        	\arrow[from=3-2, to=3-3]
        \end{tikzcd}
    \end{equation}
    
    Then $D^{\text{gauged}}_\bullet = \fib(\tau)$ is given by \Cref{lemma:gauging-gadget} applied to $D^{\text{dense}}_\bullet$. With $A_\bullet = C_\bullet$, and $B = T$. We obtain the following chain map $\pi_\bullet: D^{\text{gauged}}_\bullet \to D^{\text{dense}}_\bullet$:
    \begin{equation}
        \begin{tikzcd}
        	& {C_2} & {C_1} & {C_0} \\
        	{\mathcal{G}_3} & {\mathcal{G}_2} & {\mathcal{G}_1} \\
        	{C_2} & {C_1} & {C_0} \\
        	& T
        	\arrow[from=1-2, to=1-3]
        	\arrow[from=1-2, to=2-2]
        	\arrow["{\pi_2} =  \text{id}_{C_2}"{pos=0.8}, from=1-2, to=3-1]
        	\arrow[from=1-3, to=1-4]
        	\arrow[from=1-3, to=2-3]
        	\arrow["{\pi_1} =  \text{id}_{C_1}"{pos=0.8}, from=1-3, to=3-2]
        	\arrow["{\pi_0} =  \text{id}_{C_0}"{pos=0.8}, from=1-4, to=3-3]
        	\arrow[from=2-1, to=2-2]
        	\arrow[from=2-2, to=2-3]
        	\arrow[from=3-1, to=3-2]
        	\arrow[from=3-2, to=3-3]
        	\arrow[from=3-2, to=4-2]
        \end{tikzcd}
    \end{equation}
    Where $\pi_1$ is from \Cref{lemma:gauging-gadget}.
    Finally $D^{\text{sparse}}_\bullet = \tilde{D}^{\text{gauged}}_\bullet$ from the application of \Cref{lemma:dual-gauging-gadget}, with $B_\bullet = D^{\text{gauged}}_\bullet$, $A_1 = C_1$, and $G_1 = \mathcal{G}_2$. 
    We obtain the following chain map $\rho_\bullet : D^{\text{sparse}}_\bullet \to D^{\text{gauged}}_\bullet$:
    \begin{equation}
        \begin{tikzcd}
        	&& {C_2} & {C_1^{\oplus {\omega_\dagger+1}}} & {C_0} \\
        	& {\mathcal{G}_3} & {\mathcal{G}_2} & {\mathcal{G}_1 \oplus Q_0} \\
        	& {C_2} & {C_1} & {C_0} \\
        	{\mathcal{G}_3} & {\mathcal{G}_2} & {\mathcal{G}_1}
        	\arrow[from=1-3, to=1-4]
        	\arrow[from=1-3, to=2-3]
        	\arrow["{\rho_2} =  \text{id}_{C_2}"{pos=0.8}, from=1-3, to=3-2]
        	\arrow[from=1-4, to=1-5]
        	\arrow[from=1-4, to=2-4]
        	\arrow["{\rho_1}"{pos=0.8}, from=1-4, to=3-3]
        	\arrow["{\rho_0}"{pos=0.8}, from=1-5, to=3-4]
        	\arrow[from=2-2, to=2-3]
        	\arrow[from=2-3, to=2-4]
        	\arrow[from=3-2, to=3-3]
        	\arrow[from=3-2, to=4-2]
        	\arrow[from=3-3, to=3-4]
        	\arrow[from=3-3, to=4-3]
        	\arrow[from=4-1, to=4-2]
        	\arrow[from=4-2, to=4-3]
        \end{tikzcd}
    \end{equation}
     Where $\rho_1$ is from \Cref{lemma:dual-gauging-gadget}.
    The construction sequence gives us a natural chain map:
     \begin{equation}
         D^{\text{sparse}}_\bullet \xrightarrow{\rho} D^{\text{gauged}}_\bullet \xrightarrow{\pi} D^{\text{dense}}_\bullet \xrightarrow{\gamma} C_\bullet
     \end{equation}
     So we define $\gamma_\bullet^{\text{sparse}} \coloneq  \gamma \circ \pi \circ \rho$ as a map $ \gamma_\bullet^{\text{sparse}} : D^{\text{sparse}}_\bullet \to C_\bullet$. From \Cref{lemma:gauging-gadget,lemma:dual-gauging-gadget} we have that $\gamma_\bullet^{\text{sparse}}$ is indeed $1$-sparse.

    \begin{enumerate}
        \item From \Cref{lemma:expanding-pcm}, $D^{\text{dense}}$ has the $\syst_1(C)$-threshold property relative to the $1$-sparse map $\gamma_1$. Then, \Cref{lemma:gauging-gadget} guarantees that $D^{\text{gauged}}$ preserves the $\syst_1(C)$-threshold property relative to the $1$-sparse map $\gamma_1 \circ \pi_1$. Finally, \Cref{lemma:dual-gauging-gadget} guarantees $D^{\text{sparse}}$ is $1$-expanding relative to the $1$-sparse map $\gamma_1 \circ \pi_1 \circ\rho_1$.
        
        \item From the construction sequence, we get the following series of isomorphisms, guaranteeing that $[\gamma_\bullet^{\text{sparse}} ]$ has the right image:
     \begin{align}
         H_1(D^{\text{sparse}}) &\xrightarrow[{[\rho_1]}]{\sim} H_1(D^{\text{gauged}}) \\
         &\xrightarrow[{[\pi_1]}]{\sim } H_1(D^{\text{dense}})\xrightarrow[{[\gamma_1]}]{\sim}  \measured
     \end{align}
     \item Similarly, we have:
     \begin{align}
         \syst^1(D^{\text{sparse}}) &\geq \syst^1(D^{\text{gauged}})  \\ &\geq \syst^1(D^{\text{dense}}) \geq \syst^1(C)
     \end{align}

    \item We have already established that $\gamma_\bullet^{\text{sparse}} $ is $1$-sparse.
    From \Cref{lemma:dual-gauging-gadget}, $\partial_1^{D^{\text{sparse}}}$ is $O(\omega)$-sparse.
    \item Finally, with $\omega \leq O(1)$ since the input code is sparse:
    \begin{align}
        \dim D^{\text{sparse}}_0 & \leq \dim A_0 + \dim G_0 + O( \omega_\dagger) \dim A_1 \\
        & \leq n + sn + O( s) n \\
        & \leq O(sn) \\
        &\leq O(n (k-\kappa + d \log(n)))
    \end{align} 
    \end{enumerate}
     
\end{proof}

\subsection{Open Questions}

We have some reasons to think that the scaling given by \Cref{theorem:main} is close to optimal. On one hand (non-fast) parallel surgery can be achieved with $\tilde{O}(nk)$ overhead \cite{cowtan2025parallel}, and requires at least $\Omega(k^2)$ overhead \cite{baspin2026lowerboundoverheadsurgery}. On the other, the logical cocyles use to build $\sigma_{\text{cocycles}}$ all have support at least $d$, suggesting the $d \log(n)$ term is also unavoidable. We further note that for constant rate codes our construction is indeed optimal, up to polylogarithmic factors. 

It also remains open how to obtain sparse metachecks. In our construction, the limiting factor is currently the data layers used in \Cref{lemma:dual-gauging-gadget}, as they impose that the modified boundary map $\tilde{\partial}_2^A$ acts densely on those layers. In essence we are still lacking an expansion-preserving gadget that can reduce the degree of $Z$ checks while not increasing the $X$ weight.


\section{Sparsification gadgets}

\subsection{Gauging gadget (weight reducing)}
\label{subsec:gauging}

The following amounts to an exposition of the non-decongested gadgets from \cite{williamson2024gauging} from a chain map perspective. We only need to also show that those gadgets are expansion-preserving.

\begin{lemma}\label{lemma:gauging-gadget}
    Let
    \begin{equation}
        \begin{tikzcd}
        	{A_2} & {A_1} & {A_0} \\
        	& T
        	\arrow[from=1-1, to=1-2]
        	\arrow[from=1-2, to=1-3]
        	\arrow["\sigma", from=1-2, to=2-2]
        \end{tikzcd}
    \end{equation}
    with $\sigma \in \F_2^{s \times n}$
    Then there exists $\mathcal{G}_\bullet$, and $\tau_\bullet$:
    \begin{equation}
        \begin{tikzcd}
        	& {A_2} & {A_1} & {A_0} \\
        	{\mathcal{G}_3} & {\mathcal{G}_2} & {\mathcal{G}_1}
        	\arrow[from=1-2, to=1-3]
        	\arrow["{\tau_2}", from=1-2, to=2-2]
        	\arrow[from=1-3, to=1-4]
        	\arrow["{\tau_1}", from=1-3, to=2-3]
        	\arrow[from=2-1, to=2-2]
        	\arrow[from=2-2, to=2-3]
        \end{tikzcd}
    \end{equation}
    Such that $\fib(\tau)$ satisfies the following properties:
    \begin{enumerate}
        \item Given $\gamma_\bullet$ the natural projection from $\fib(\sigma)$ to $A_\bullet$ such that $\fib(\sigma)$ has the $d$-threshold property relative to $\gamma_1$. 
        
        Then there exists a map $\pi_\bullet : \fib(\tau) \to \fib(\sigma)$ such that $\gamma_\bullet \circ \pi_\bullet$ is sparse; and $\fib(\tau)$ has the $d$ threshold property relative to $\gamma_1 \circ \pi_1$. 
        \item  The projection $\pi_1$ induces the isomorphism $[\pi_1]: H_1(\fib(\tau)) \xrightarrow{\sim} H_1(\fib(\sigma))$
        \item $\syst^1(\fib(\tau)) \geq \syst^1(\fib(\sigma))$
        \item The map $\tau_1$ is $O(s)$-bounded, while $(\partial_1^{\mathcal{G}})$ $O(1)$-sparse, and $\tau_1^T$ is $1$-bounded
        \item $\dim \mathcal{G}_1 \leq sn$ 
    \end{enumerate}
\end{lemma}

    We begin by writing $\sigma$ as a matrix of $s$ row vectors, each with length $n$:
    \begin{equation}
        \sigma = \begin{bmatrix}
            \quad & \vec{\sigma}_1 & \quad \\ 
            \quad & \vec \sigma_2 & \quad\\
            \quad & \vdots \\
            \quad & \vec \sigma_s &\quad
        \end{bmatrix} 
    \end{equation}

    Then to each of those vectors, we associate a connected expander graph $G_j = (V_j, E_j)$, with $|V_j| = |\supp(\sigma_j)|$ from \Cref{lemma:constant-degree-expanders}, along with an injection $\iota_j: \supp(\sigma_j) \to V_j$. The graphs can each be augmented with faces $F_j$ such that
    \begin{equation}
        H_2(\mathcal{G}^j) = 0, \quad \mathcal{G}^j = \F_2[F_j] \to \F_2[E_j] \to \F_2[V_j]
    \end{equation}
    We then define a chain map $\tau^j: A_\bullet \to \mathcal{G}^j$:
    \begin{equation}
        \tau^j_1(a_1) = \iota(a_1), \quad \tau = \bigoplus_j \tau^j
    \end{equation}
    To define $\tau_2$, we make the observation that $\partial a_2 \in A_1$ has even overlap with all $\vec{\sigma}_j$'s, so we can group this support into disjoint pairs: $\supp(\partial a_2) \cap \supp(\vec{\sigma}_j)\cong \{ (a_l, a_{l+1})\}_{l=1}^{M}$, for some $M \in \N$. As $G_j$ is connected, for each of those pairs, there exists a path $p_j((a,b)) \subset E_j$ in $G_j$, with startpoint $a$, and endpoint $b$. We thus write
    \begin{equation}
        \tau^j_2(a_2) = \bigoplus_j \sum_{(a,b) \in\supp(\partial a_2) \cap \supp(\vec{\sigma}_j) } p_j((a, b))
    \end{equation}
    This gives us the following chain map:
    \begin{equation}
        \begin{tikzcd}
        	& {A_2} & {A_1} & {A_0} \\
        	{\F_2[F_j]} & {\F_2[E_j]} & {\F_2[V_j]}
        	\arrow[from=1-2, to=1-3]
        	\arrow["{\tau_2^j}", from=1-2, to=2-2]
        	\arrow[from=1-3, to=1-4]
        	\arrow["{\tau_1^j}", from=1-3, to=2-3]
        	\arrow[from=2-1, to=2-2]
        	\arrow[from=2-2, to=2-3]
        \end{tikzcd}
    \end{equation}
    
    We then sum over all the individual graph gadgets. Writing $\mathcal{G} = \bigoplus_j \mathcal{G}^j$, $\tau_\bullet = \bigoplus_j \tau^j_\bullet $, we have:
    \begin{equation}
        \begin{tikzcd}
        	& {A_2} & {A_1} & {A_0} \\
        	{\mathcal{G}_3} & {\mathcal{G}_2} & {\mathcal{G}_1}
        	\arrow[from=1-2, to=1-3]
        	\arrow["{\tau_2}", from=1-2, to=2-2]
        	\arrow[from=1-3, to=1-4]
        	\arrow["{\tau_1}", from=1-3, to=2-3]
        	\arrow[from=2-1, to=2-2]
        	\arrow[from=2-2, to=2-3]
        \end{tikzcd}
    \end{equation}

    \begin{remark}
    Note that we do not use the decongested version of the construction of \cite{williamson2024gauging}, as we are not concerned with the boundedness of the resulting $\partial_2^{\fib(\tau)}$.
    \end{remark}

    \begin{proofof}{\Cref{lemma:gauging-gadget}}
        \begin{enumerate}

            \item 

            Intuitively, that property is preserved simply because the graph gadgets, through $\tau$, preserve the parity of violations of row of $\sigma$, and therefore the weight of the syndrome. The "parity of violations" for the gadgets can be studied formally by focusing on $\tau^*: A_1 \to H_1(\mathcal{G})$ \footnote{Note that $H_1(\mathcal{G}) \cong \F_2^s \cong T$, this hints at $\tau$ preserving the syndrome information given by $\sigma$.}. 
            
            From the definition of $\tau^j_1$, we read that $\tau^j_1 \ket{a_1} = \ket{\iota(a_1)} \braket{\sigma_j | a_1}$, so $|\supp(\tau^j_1 \ket{\vec{a}})|$ has the same parity as $\braket{\sigma_j|a}$. 
            From $G_j$ being connected, all vectors of the same parity are the same class in $H_1(\mathcal{G}^j)$, hence $(\tau^j_1)^* \cong \sigma_j$; where we write $(\tau^j_1)^* : A_1 \to H_1(\mathcal{G}^j)$. 
            By linearity we obtain $\tau_1^* \cong \sigma$. 
            
            Consequently, we have $\ker \tau_1^* = \ker \sigma$, which allows us to use the natural projection $\pi_1: A_1 \oplus \mathcal{G}_2 \to A_1$, under which the kernels are preserved:
            \begin{align}
                Z_1(\fib(\sigma)) = \pi_1 Z_1(\fib(\tau))
            \end{align}
            The chain map $\pi_\bullet$ is then given by:
            \begin{align}
                \pi_2 &: A_2 \oplus \mathcal{G}_3 \to A_2, \quad \pi_2(a_2 \oplus g_3) = a_2 \\
                \pi_1&: A_1 \oplus \mathcal{G}_2 \to A_1, \quad \pi_1 (a_1 \oplus g_2)= a_1 \\
                \pi_0&: A_0 \oplus \mathcal{G}_1 \to A_0 \oplus H_0(\mathcal{G}) \cong T \\  
                \pi_0& \ket{a, v} = \ket{a, [v]}
            \end{align}
            All of the maps above are trivially $1$-sparse except $\pi_0$, for which we can verify $\gamma_0 (T) = 0$, i.e. $\gamma_0 \circ \pi_0 = \text{id}_{A_0}$ which is $1$-sparse.
            
            We then move on to lower bounding the syndrome. First, we trivially have:
            \begin{align}
                &\norm{\partial_1^{\fib(\tau)} \ket{{a}, {e}}} \\
                = &\norm{\tau_1 \ket{{a}} + \partial_2^{\mathcal{G}}\ket{{e}}}  + \norm{\partial_1^{A} \ket a} \\
                = &\norm{\tau_1 \ket{{a}} + \partial_2^{\mathcal{G}}\ket{{e}}}  + \norm{\partial_1^{A} \pi_1 \ket{a,e}}
            \end{align}
            While $\pi_0$ itself is not sparse the composition $\gamma_0 \circ \pi_0$ is, as it eliminates the $T$ component entirely.
            
            Now we break down the vector $\ket{{e}}$ into graph-wise components, i.e. $\ket{{e}} = \bigoplus_j \ket{e_j}$. This gives 

            \begin{align}
                &\norm{\tau_1 \ket{{a}} + \partial^{\mathcal{G}}\ket{{e}}} \\ =& \norm{\bigoplus_j \tau^j_1 \ket{{a}} + \partial^{\mathcal{G}_j}\ket{e_j}} \\
                \geq & \norm{\bigoplus_j (\tau^j)^* \ket{a}} \cdot \min_j \syst_1(\mathcal{G}^j) \\
                \geq& \norm{\bigoplus_j (\tau^j)^* \ket{a}} =  \norm{\tau^* \ket{a}} \\
                 = & \norm{\sigma \ket{a}}\\
                 = & \norm{\sigma \pi_1 \ket{a, e}}\\
            \end{align}
            Hence 
            \begin{equation}
                \norm{\partial_1^{\fib(\tau)} \ket{{a}, {e}}}  \geq \norm{\partial_1^{\fib(\sigma) } \pi_1 \ket{{a}, {e}}} 
            \end{equation}

            Thus we get that if
            \begin{equation}
                (\gamma_1 \circ \pi_1 ) \ket{a, e} \not \in (\gamma_1 \circ \pi_1) Z_1(\fib(\tau))
            \end{equation}
            Or equivalently, 
            \begin{equation}
                \gamma_1 (\pi_1  \ket{a, e} ) \not \in \gamma_1 (\pi_1 Z_1(\fib(\tau))) = \gamma_1 Z_1(\fib(\sigma)))
            \end{equation}
            Then we get:
            \begin{equation}
                 \norm{\partial_1^{\fib(\tau)} \ket{{a}, {e}}}  \geq \norm{\partial_1^{\fib(\sigma) } \pi_1 \ket{{a}, {e}}} \geq d
            \end{equation}
            
            \item We already have $\tau_1^* \cong \sigma$, hence $\ker [\tau] = \ker[\sigma]$. 
            Since $H_2(\mathcal{G}) = 0$ from assumptions, we can then compute the homology group of the mapping cone $H_1(\fib(\tau)) = H_1(\fib(\sigma))$ \cite{weibel1994introduction, baspin2025fastsurgeryquantumldpc}.
            
            \item $\syst^1(\fib(\tau)) \geq \syst^1(\fib(\sigma))$ is preserved from expansion, see \cite{williamson2024gauging}, and also \cite{baspin2025fastsurgeryquantumldpc, yuan2026quantumweightreductionlayer}.

            \item By definition, the $\partial^{\mathcal{G}}_2$ is $2$-bounded, and $(\partial^{\mathcal{G}}_2)^T$ is $d_0$-bounded, hence $\partial^{\mathcal{G}}_2$ is $O(d_0)$-sparse, or $O(1)$-sparse. 

            On the other hand, $\tau_1$ is $s$-bounded, and $\tau^T_1$ is $1$-bounded.
            

            \item $\dim \mathcal{G}_2 \leq s \cdot \max_j |V_j| \leq sn$ by construction. 
        \end{enumerate}
    \end{proofof}

\begin{remark}
    While we emphasised the mapping fiber picture for the gauging gadgets to simplify proofs, we also briefly sketch how it can be seen as an instance of a code embedding \cite{yuan2026unified}. The point of this remark is to highlight the duality between the gauging gadget, and the dual gauging gadget, which we will express as a generalised layer code. 
    We first write the complex to be sparsified in an embedding-friendly fashion.
    \begin{equation}
        \fib(\sigma) = \begin{tikzcd}
        	{A_2} & {A_1} & {T \oplus A_0}
        	\arrow["{\partial_2^A}", from=1-1, to=1-2]
        	\arrow["\begin{array}{c} \begin{bmatrix} \sigma \\ \partial^A_1 \end{bmatrix} \end{array}", from=1-2, to=1-3]
        \end{tikzcd}
    \end{equation}
    We can modify the gauging construction slightly to use a single graph $G$ on $n$ vertices, without changing the rest of the construction. The height-2 cone corresponding to the code embedding can then be written as follows:
    \begin{align}
        \begin{tikzcd}[ampersand replacement=\&]
        	\&\& {A_2} \\
        	\&\& {A_1} \\
        	{T\otimes \F_2[F]} \& {T \otimes \F_2[E]} \& \begin{array}{c} \begin{matrix} T \otimes \F_2[V] \\\oplus A_0 \end{matrix} \end{array}
        	\arrow["{\partial_2^A}", from=1-3, to=2-3]
        	\arrow["{p = \tau_2}"', from=1-3, to=3-2]
        	\arrow["{ \begin{bmatrix}\tau_1 \\ \partial^A_1 \end{bmatrix} }", from=2-3, to=3-3]
        	\arrow[from=3-1, to=3-2]
        	\arrow[from=3-2, to=3-3]
        \end{tikzcd}
    \end{align}
    This corresponds to replacing \emph{some} of the checks with some larger, sparser, check layers; namely checks are replaced with graph vertices.
    As $H_1(\mathcal{G}) = 0$, the resulting embedded complex then corresponds to the target code:
    \begin{equation}
        A'^{\textbf{g}} \cong \fib(\sigma)
    \end{equation}
    This rephrasing might be useful for the reader, as in the next section we will instead replace qubits with data layers; namely qubits will be replaced with graph vertices.
\end{remark}

\subsection{Dual-gauging gadget (degree reducing)}
\label{subsec:dual-gauging}

In the following section we describe how to reduce the degree of the surgery ancilla. The gadget can essentially be described as a generalised layer code \cite{yuan2026quantumweightreductionlayer, yuan2026unified} where the qubits with very high degree are replaced by clusters of vertices from an expander graph, instead of the usual surface code layer.

\begin{lemma}[Dual-gauging gadget] \label{lemma:dual-gauging-gadget}
    Let
    \begin{equation}
        B_\bullet = \begin{tikzcd}[ampersand replacement=\&]
	{B_2} \& {G_1 \oplus A_1} \&\& {G_0 \oplus A_0}
	\arrow["{\left[\begin{smallmatrix}
            \partial_2^G \\
            \partial_2^A
        \end{smallmatrix}\right]}", from=1-1, to=1-2]
	\arrow["{\left[\begin{smallmatrix}
            \partial_1^G & \tau_1 \\
            0 & \partial_1^A
        \end{smallmatrix}\right]}", from=1-2, to=1-4]
\end{tikzcd}
    \end{equation}
    Such that $\partial_1^A, \partial_1^G$ are $\omega$-sparse, while  $\tau_1$ is $\omega_\dagger$-bounded and $(\tau_1)^T$ is $\omega$-bounded. 
    Then there exists output code $\tilde{B}_\bullet$:

    \begin{equation}
        \tilde{B}_\bullet = \begin{tikzcd}
        	{\tilde{B}_2} & {G_1 \oplus A_1^{\oplus \omega_\dagger + 1}} & {G_0 \oplus A_0 \oplus Q_0}
        	\arrow["\partial_2^{\tilde{B}}", from=1-1, to=1-2]
        	\arrow["\partial_1^{\tilde{B}}", from=1-2, to=1-3]
        \end{tikzcd}
    \end{equation}
    
    Such that:
    \begin{enumerate}
        \item 
        Write the reduced complex $A'_\bullet$:
        \begin{equation}
            A_\bullet'  = B_2 \to A_1 \to A_0
        \end{equation}
        There exists a natural map $\tilde{\gamma}_\bullet: B_\bullet \to A_\bullet'$:
        \begin{align}
            \tilde{\gamma}_2&: B_2 \to B_2, \quad \tilde\gamma_2 = \text{id}_{B_2} \\
            \tilde{\gamma}_1&: G_1 \oplus A_1 \to A_1, \quad \tilde\gamma_1(g_1 \oplus a_1) = a_1\\
            \tilde{\gamma}_0&: G_0 \oplus A_0\to A_0, \quad \tilde\gamma_0(g_0 \oplus a_0) = a_0
        \end{align}
        If $B_\bullet$ has the $d$-threshold property relative to $\tilde{\gamma}_1$ for $d = \syst_1(A')$, then there exists a map $\rho_\bullet : \tilde{B}_\bullet \to B_\bullet$ such that $\tilde{\gamma}_\bullet \circ \rho_\bullet$ is sparse; and $\tilde{B}_\bullet$ is $1$-systolic expanding relative to $\tilde{\gamma}_1 \circ \rho_1$.
        \item 
        The map $\rho_1:  \tilde{B}_1\to B_1$ induces an isomorphism $[\rho_1]: H_1(\tilde{B}) \xrightarrow{\sim} H_1(B)$
        \item $\syst^1(\tilde{B}) \geq \syst^1(B)$
        \item $\partial_1^{\tilde{B}}$ is $O(\omega)$-sparse
        \item $\dim \tilde{B}_0 = \dim A_0 + \dim G_0 +  O(\omega_\dagger)\dim A_1$ 
    \end{enumerate}
\end{lemma}

As previously mentioned, our approach relies on replacing the qubits in $A_1$ by some data layers \cite{yuan2026quantumweightreductionlayer, yuan2026unified, williamson2024layer}. 
Using \Cref{lemma:constant-degree-expanders} pick $X = (V,E)$ an $(\omega + 2)$-expander graph on $\omega_\dagger + 1$ vertices, then each qubit in $A_1$ is replaced by a cluster $V$. 
A number $\omega_\dagger$ of those vertices will be used to sparsify $\tau_1$, while the last one will be used for the already sparse $\partial_1^A$. The data layers then correspond to the complex $Q_\bullet$.
\begin{align}
    Q_\bullet =& Q_1 \xrightarrow{\partial^Q_1} Q_0 \\
    =& A_1 \otimes \F_2[V] \xrightarrow{\text{id} \otimes \partial^V} A_1 \otimes \F_2[E]
\end{align}
Note that $Q_1 = A_1^{\oplus \omega_\dagger + 1}$. For convenience, we will use the notation $V = \{ 0, 1, \dots, \omega_\dagger\}$, and $\F_2[V] = \spn\{ \ket{0}, \ket{1}, \dots, \ket{\omega_\dagger}\}$.

Then, the output code $\tilde{B}_\bullet$ is defined as the code embedding obtained from the following height-2 cone 
\begin{equation}
    \begin{tikzcd}
        {B_2} & \\
        {G_1 \oplus Q_1} & {Q_0} \\
        {G_0 \oplus A_0}
        \arrow["f", from=1-1, to=2-1]
        \arrow["{\partial^Q}", from=2-1, to=2-2]
        \arrow["{f'}", from=2-1, to=3-1]
    \end{tikzcd}
\end{equation}
This means we have:
\begin{equation}
    \tilde{B}_\bullet = \begin{tikzcd}
	{B_2} & {G_1 \oplus Q_1} && {G_0 \oplus A_0 \oplus Q_0}
	\arrow["f", from=1-1, to=1-2]
	\arrow["\begin{array}{c} \begin{bmatrix}f' \\ 0 \ \partial^Q \end{bmatrix} \end{array}", from=1-2, to=1-4]
\end{tikzcd}
\end{equation}
The explicit form of the maps above is given as follows:
\begin{equation}
    \begin{tikzcd}[ampersand replacement=\&]
	{B_2} \&\& \\
	\\
	{G_1 \oplus A_1 \otimes \F_2[V]} \&\& {A_1 \otimes \F_2[E]} \\
	\\
	{G_0 \oplus A_0}
	\arrow["{\begin{bmatrix} \partial_2^G \\ \tilde{\partial}_2^{A}  \end{bmatrix}  }"', from=1-1, to=3-1]
	\arrow["{\text{id} \otimes \partial^V}", from=3-1, to=3-3]
	\arrow["{\begin{bmatrix} \partial_1^G & \tilde{\tau}_1 \\ 0 & \tilde{\partial}_1^A \end{bmatrix}}"', from=3-1, to=5-1]
\end{tikzcd}
\end{equation}
Where we write
\begin{equation}
    \tilde{\partial}_2^A = \partial_2^A \otimes \ket{\vec{1}_V}, \qquad \tilde{\partial}_1^A = \partial_1^A \otimes \bra{v_\star}.
\end{equation}
Where $v_\star$ is a distinguished vertex, and we write $\ket{\vec{1}_V} = \sum_{v \in V} \ket{v}$. Thus the original syndrome information from $A_0$ is carried by the distinguished copy $a_1\otimes \ket{v_\star}$.
It remains to sparsify $\tau_1$. Since $\tau_1$ is $\omega_\dagger$-bounded, then for any basis element $a_1 \in A_1$:
\begin{equation}
    \card{\tau_1 (a_1)} \leq \omega_\dagger.
\end{equation}
Choosing $|V|=  \omega_\dagger+1$, we may therefore pick, for each
basis element $a_1$, an injective labeling
\begin{equation}
    \iota_{a_1}: \supp(\tau_1 (a_1)) \to V\setminus\{v_\star\}.
\end{equation}
We then define the twisted map $\tilde{\tau}_1$ by
\begin{equation}
    \tilde{\tau}_1 (a_1\otimes \ket{v}) = \sum_{\substack{g_0 \in \supp(\tau_1 a_1) \\ \iota_{a_1}(g_0) = v}} g_0 
\end{equation}

\begin{proofof}{\Cref{lemma:dual-gauging-gadget}}
    \begin{enumerate}

        \item 
        We start by definition the chain map $\rho_\bullet$. Naturally, we take $\rho_2: \tilde{B}_2 = B_2 \to B_2$ as the identity on $B_2$. 
        
        We then define $\rho_1: G_1 \oplus A_1 \otimes \F_2[V] \to G_1 \oplus A_1$ as the projection on the distinguished vertex $v_\star\in V$. 
        \begin{equation}
            \rho_1 = \text{id}_G \oplus \text{id}_A \otimes \bra{v_\star}
        \end{equation}

        For $\rho_0$, we begin by assuming a map of the form $\rho_0 (g_0 \oplus a_0\oplus e_0) = (g_0 + \phi(e_0)) \oplus a_0$, for some function $\phi: A_1 \otimes \F_2[E] \to G_0$. Through the commutation condition, we get this map has to obey the following relation:
        \begin{align}
            \label{eq:phi-commutation}
            &\phi \circ (\text{id} \otimes \partial^V) (a \otimes \ket{v} )= \nonumber \\ 
            & \qquad \tau_1 a  \braket{v_\star | v} - \tilde{\tau}_1 (a \otimes \ket{v})
        \end{align}
        We see here that the role of $\phi$ is to resolve the incongruence between $\tau_1 \circ \rho_1$ computing some syndrome reading values only from $\F_2[v_\star]$, while $\tilde{\tau}_1$ computes the syndrome from the entirety of $\F_2[V]$. 

        The simplest solution is to add some correction to the syndrome in $G_0$ whenever the values of $v$ and $v_\star$ differ. From the perspective of $\phi$ this can be accomplished by, along a path from $v \in V$ to $v_\star \in V$ computing the parity of the edges encountered. Write $\ket{\Phi_{v \rightsquigarrow v_\star }} \in \F_2[E]$ a choice of path from $v$ to $v_\star$, and $\ket{\nu} \in \F_2[V]$ an arbitrary state on the vertex space. Then:
        \begin{equation}
            \bra{\Phi_{v \rightsquigarrow v_\star }} \partial^V \ket{\nu} = \braket{v|\nu} - \braket{v_\star|\nu}
        \end{equation}
        Summing over $v$'s we get:
        \begin{align}
            (&\sum_v \ket{v}\bra{\Phi_{v \rightsquigarrow v_\star }}) \partial^V \ket{\nu} \\
            = &\sum_v \ket{v} (\braket{v|\nu} - \braket{v_\star|\nu}) \\
            = &\ket{\nu} - \braket{v_\star|\nu} \ket{\vec{1}_V}
        \end{align}
        Where $\ket{\nu} - \braket{v_\star|\nu} \ket{\vec{1}_V}$ is a vector with value $0$ at $v$ if $\braket{v|\nu} =\braket{v_\star|\nu}$ and $1$ otherwise: it encodes locations that need correction relative to $v_\star$.
        We can then write:
        \begin{gather}
        \Pi: \F_2[E] \to \F_2[V]
        \Pi = \sum_v \ket{v}\bra{\Phi_{v \rightsquigarrow v_\star }} \\
            \phi = \tilde{\tau}_1 \circ (\text{id}_A \otimes \Pi)
        \end{gather}
        And one can readily verify that $\phi$ satisfies equation \Cref{eq:phi-commutation}.
        Note that resulting $\rho_0$ is not a sparse map, due to $\Pi$ computing some parity along potentially long paths. However, since $\tilde{\gamma}_0 (G_0)= 0$, and therefore $\tilde{\gamma}_0  \circ \phi = 0$ we still have $\tilde{\gamma}_0\circ \rho_0$ $1$-sparse. 
        
        We now move on to lower bound the expansion. Let $\tilde{b} = g\oplus \tilde{a} \in G_1 \oplus A_1 \otimes \F_2[V]$. We write
        \begin{equation}
            \tilde{a}  = \bigoplus_{a} a\otimes \ket{\nu_a}, \quad \ket{\nu_a} \in \F_2[V]
        \end{equation}
        Then for any $\tilde{a}$ there exists a correction $\tilde{a}_{\text{corr}} = \bigoplus_{a} a\otimes \ket{\nu_a^{\text{corr}}}$ with $\ket{\nu_a^{\text{corr}}} \in \F_2[V]$ such that $\ket{\nu_a}+ \ket{\nu_a^{\text{corr}}} \in \ker \partial^V$, and we can pick $\norm{\ket{\nu_a^{\text{corr}}}} \leq |V|/2$.
        
        We write $\tilde{a}_{\text{dec}} = \tilde{a} + \tilde{a}_{\text{corr}}$ the "decoded" version of $\tilde{a}$. As $\ker\partial^V = \spn\{\ket{\vec{1}_V}\}$, there exists $a_{\text{dec}}$ such that 
        \begin{equation}
            \tilde{a}_{\text{dec}} = a_{\text{dec}} \otimes \ket{\vec{1}_V}
        \end{equation}
        Correspondingly, we write $\tilde{b}_{\text{dec}} = g \oplus \tilde{a}_{\text{dec}}$, and $b_{\text{dec}} = g \oplus a_{\text{dec}}$. Crucially, we have:
        \begin{align}
            \partial_1^{\tilde{B}} \tilde{b}_{\text{dec}} = & f'(\tilde{b}_{\text{dec}}) \\
            = & (\partial_1^G g + \tilde{\tau}_1\tilde{a}_{\text{dec}}) \oplus \tilde{\partial}_1^A \tilde{a}_{\text{dec}} \\ 
            =& (\partial_1^G g + \tau_1 a_{\text{dec}}) \oplus \partial_1^Aa_{\text{dec}} \\
            =& \partial_1^B b_{\text{dec}}
        \end{align}

        We can then compute:
        \begin{align*}
            \norm{\partial_1^{\tilde{B}} \tilde{b}} = & \norm{\partial_1^{\tilde{B}} (\tilde{b}_{\text{dec}} + \tilde{a}_{\text{corr}})} \\
            =  & \norm{f'(\tilde{b}_{\text{dec}}+ \tilde{a}_{\text{corr}} )} + \norm{\text{id} \otimes\partial^V\tilde{a}_{\text{corr}}} \\
            \geq &  \norm{f'(\tilde{b}_{\text{dec}})} - \norm{ \tilde{\tau}_1(\tilde{a}_{\text{corr}})} \\& - \norm{ \tilde{\partial}_1^A(\tilde{a}_{\text{corr}})}+  (\omega+2)\norm{\tilde{a}_{\text{corr}}} \\
            \geq & \norm{\partial_1^B(b_{\text{dec}})} - \norm{\tilde{a}_{\text{corr}}} \\ &- \omega \norm{\tilde{a}_{\text{corr}}} +  (\omega + 2)\norm{\tilde{a}_{\text{corr}}} \\
            \geq &\norm{\partial_1^B(b_{\text{dec}})} +\norm{\tilde{a}_{\text{corr}}}
        \end{align*}

        If $\tilde{\gamma}_1(b_{\text{dec}}) \not \in \tilde{\gamma}_1(Z_1(B))$, then $\norm{\partial_1^B(b_{\text{dec}})} \geq d$ and thus $\norm{\partial_1^{\tilde{B}} \tilde{b}}  \geq d$. 
        Now suppose that $\tilde{\gamma}_1(b_{\text{dec}}) \in \tilde{\gamma}_1(Z_1(B))$.  This means that $\tilde{\gamma}_1 \circ \rho_1 (\tilde{b}_\text{dec}) \in \tilde{\gamma}_1 \circ \rho_1 Z_1(\tilde{B})$. In which case then 
        \begin{align}
            &\dist(\tilde{\gamma}_1 \circ \rho_1 (\tilde{b}), \tilde{\gamma}_1 \circ \rho_1 Z_1(\tilde{B})) \\ 
            & \leq \dist(\tilde{\gamma}_1 \circ \rho_1 ( \tilde{b}), \tilde{\gamma}_1 \circ \rho_1 (\tilde{b}_\text{dec}) ) \\
            & \leq \norm{\tilde{\gamma}_1 \circ \rho_1 ( \tilde{a}_\text{corr})} \\
            & \leq \norm{\tilde{a}_{\text{corr}}} \\
            & \leq \norm{\partial_1^{\tilde{B}} \tilde{b}}
        \end{align}

        Thus $\tilde{B}_\bullet$ is $1$-systolic expanding relative to $\tilde{\gamma}_1 \circ \rho_1$.
        
        \item We have $Z_1(\tilde{B}) \subset \ker(\text{id} \otimes \partial^V) = G_1 \oplus A_1 \otimes \ket{\vec{1}_V}$. On that space, we have $\tilde{\tau}_1 (a \otimes \ket{ \vec{1}}_V) = \tau_1 (a)$ and $\tilde{\partial}_1^A (a \otimes \ket{ \vec{1}}_V) = \partial_1^A a$.
        Hence, we have $\rho_1 Z_1(\tilde{B}) = Z_1(B)$, and $\rho_1 B_1(\tilde{B}) = B_1(B)$.
        
        This shows that the projection $\rho_1$ induces the isomorphism $[\rho_1]: H_1(\tilde{B}) \xrightarrow{\sim} H_1(B)$.
        
        \item We can compute $H^1(\tilde{B})$ from the following embedded co-complex \cite{yuan2026unified}:
        \begin{equation}
            B_2 \leftarrow G_1 \oplus H^1(Q) \leftarrow B_0
        \end{equation}

        As previously noted, the complex $\F_2[V] \to \F_2[E]$ corresponds to a connected graph. Therefore the coboundaries induced by $\F_2[E]$ connect any two vertices, giving $H^1(\mathcal{G}) \cong \F_2$. 
        Here this gives $H^1(Q) = A_1 \otimes H^1(\mathcal{G}) \cong A_1 \otimes \F_2 \cong A_1$. 
        The coboundary maps $(\tilde{\partial}_2^A)^T, (\tilde{\partial}_1^A)^T, (\tilde{\tau}_1)^T$ behave as $(\partial_2^A)^T, (\partial_1^A)^T, (\tau_1)^T$ on $H^1(Q) \cong A_1$; 
        and therefore $(\partial_2^{\tilde{B}})^T, (\partial_1^{\tilde{B}})^T$ act as $(\partial_2^{B})^T, (\partial_1^{B})^T$ on $G_1 \oplus H^1(Q)$, giving:
        \begin{equation}
            \syst^1(\tilde{B}) \geq \syst^1(B)
        \end{equation}

        \item We have that $\tilde{\tau}_1$ is $1$-bounded and $\tilde{\partial}_1^A$ is $\omega$-bounded, while $\text{id} \otimes \partial^V$ is $d_0 \lceil  \frac{\omega+2}{h_0}\rceil$-bounded, or $O(\omega)$-bounded, and $\partial_1^G$ is $\omega$-bounded, hence $\partial_1^{\tilde{B}}$ is $O(\omega)$-bounded.
        
        On the other hand, we have that $(\tilde{\tau}_1)^T$ is $\omega$-bounded and $(\tilde{\partial}_1^A)^T$ is $\omega$-bounded too, while the transpose of $\text{id} \otimes \partial^V$ is $2$-bounded, and $(\partial_1^G)^T$ is $\omega$-bounded, hence $(\partial_1^{\tilde{B}})^T$ is $O(\omega)$-bounded.
        
        \item By definition, $\dim Q_0 = O(\omega \omega_\dagger) \dim A_1$, hence $\dim \tilde{B}_0 \leq \dim A_0 + \dim G_0 + O( \omega\omega_\dagger) \dim A_1$
    \end{enumerate}
\end{proofof}


\section{Auxillary lemmas}

\begin{lemma}
\label{lemma:constant-degree-expanders}There exist constants $d_0, h_0$ such that for every $\lambda >0$ there exists a family $\{X_m\}_{m\ge 1}$ of $d_0 \lceil \frac{\lambda}{h_0}\rceil$-regular graphs $X_m$ on $m$ vertices with 
\begin{equation}
    h(X_m) = \min_{\substack{ 0<|S|\le |V(X_m)|/2}} \frac{|E(S,V(X_m)\setminus S)|}{|S|} \geq \lambda
\end{equation}
\end{lemma}
\begin{proof}
    It suffices to take $\lceil \frac{\lambda}{h_0}\rceil$ copies of the graphs from Lemma 2.1 of \cite{dinur2007pcp}.
\end{proof}

\begin{lemma}[Random coding boosting]
    \label{lemma:random-coding-boosting}
    Let $P: \F_2^n \to \F_2^m$, and
    \begin{equation}
        S_d = \{ c \in \im P: 0< \norm{c} < d\}
    \end{equation}
    there exists $R: \mathbb{F}_2^m \to \mathbb{F}_2^t$, such that
    \begin{equation}
        \forall c \in S_d, \quad  \norm{R c} \geq d
    \end{equation}
    With
    \begin{equation}
        t \leq O\left(d \log(em/d)\right)
    \end{equation}
\end{lemma}
\begin{proof}
    We closely follow the proof of Theorem 4.2.1 in \cite{guruswami2025essential}.
    Pick a random matrix $R \in \F_2^{t \times m}$. Then for any $c \in \F_2^m$ the vector $R c$ is random, hence:
    \begin{align}
        \Pr[\norm{Rc} < d] \leq \frac{\{\text{\# of $\F_2^t$ vectors of weight $ \leq d-1 $}\}}{\{\text{\# of vectors in $\F_2^t$}\}} 
    \end{align}
    The numerator is controlled by a sum of binomial coefficients
    \begin{equation}
        \{\text{\# of $\F_2^t$ vectors of weight $ \leq d-1 $}\} \leq \sum_{i=0}^{i = d-1} \binom{t}{i}
    \end{equation}
    With $H_2$ the binary entropy function, we have the following upper bound:
    \begin{equation}
        \sum_{i=0}^{i = d-1} \binom{t}{i} \leq 2^{t H_2((d-1)/t)}
    \end{equation}
    As we can always take $t \geq 4(d-1)$ as it does not contradict the desired result, we have $H_2((d-1)/t)$ an upper bounded constant, which we write $ \epsilon = H_2((d-1)/t) < 1$.
    Using a union bound we get:
    \begin{equation}
        \label{eq:probability-expansion}
        \Pr[\exists c\in S_d: \norm{Rc} < d] \leq \card{S_d} \cdot 2^{-(1-\epsilon) t}
    \end{equation}
    As long as $\card{S_d} < 2^{(1-\epsilon) t}$ then \Cref{eq:probability-expansion} gives that the probability of failure is less than $1$, and therefore there exists a suitable matrix $R$. For that condition to be satisfied it is sufficient to take
    \begin{equation}
        (1-\epsilon)t > \log\card{S_d}
    \end{equation}
    From \cite{jukna2011extremal} Proposition 1.4 we have the following upper bound:
    \begin{equation}
    \card{S_d} \leq \sum_{i=0}^{i=d-1} \binom{m}{i} \leq \left(\frac{e m}{d}\right)^d
    \end{equation}
    We then choose, for some sufficiently large constants:
    \begin{equation}
        t = \Theta(d \log(em/d))
    \end{equation}
    
\end{proof}

\bibliography{refs}
\end{document}